\documentclass[11pt]{article}

\newcommand{\pa}{\mathbf{PA}}
\newcommand{\aur}{\text{\textsf{AUR}}}
\newcommand{\pamax}{\textbf{PA}$_{\max}$}
\newcommand{\uij}{u_{ij}}
\newcommand{\xij}{x_{ij}}
\newcommand{\us}{u^*}
\newcommand{\Qs}{Q^*}
\newcommand{\Rs}{R^*}
\renewcommand{\tt}{\tilde{t}}

\newcommand{\tA}{\tilde{A}}
\newcommand{\vecu}{\vec{u}}
\newcommand{\vecb}{\vec{b}}
\newcommand{\veca}{\vec{a}}
\newcommand{\vecx}{\vec{x}}

\newcommand{\ucap}{\hat{u}}
\newcommand{\sumjm}{\sum_{j=1}^m}
\newcommand{\SW}{\text{\textsf{SW}}}
\newcommand{\swa}{\text{\textsf{SW}}^A}
\newcommand{\swopt}{\text{\textsf{SW}}_{\text{\textsf{OPT}}}(\vecu_1,\vecu_2)}
\newcommand{\swoptnull}{\text{\textsf{SW}}_{\text{\textsf{OPT}}}}

\usepackage{amsmath,amsthm,amsfonts,amssymb,color,hyperref,enumitem,graphicx,epstopdf}
\usepackage[usenames,dvipsnames]{xcolor}

\newcommand{\rr}{\mathbb{R}}

\newcommand{\ra}{\rightarrow}

\DeclareMathOperator*{\argmax}{arg\,max}

\newtheorem{theorem}{Theorem}

\newtheorem{prop}[theorem]{Proposition}

\usepackage[margin=0.95in]{geometry}

\title{Better Strategyproof Mechanisms without Payments or Prior\\
--- An Analytic Approach\thanks{
This work has received funding from the Vienna Science and Technology Fund (WWTF) through project ICT10-002,
and from the European Research Council under the European Union's
Seventh Framework Programme (FP7/2007-2013) / ERC Grant Agreement no.~340506.
}}
\author{Yun Kuen Cheung\\
University of Vienna, Faculty of Computer Science, Vienna, Austria\\
yun.kuen.cheung@univie.ac.at}
\date{}

\begin{document}

\maketitle

\begin{abstract}
We revisit the problem of designing strategyproof mechanisms for allocating divisible items among two agents who have linear utilities,
where payments are disallowed and there is no prior information on the agents' preferences.
The objective is to design strategyproof mechanisms which are competitive against the most efficient (but not strategyproof) mechanism.

For the case with two items:
\begin{itemize}[leftmargin=0.3cm]
\item We provide a set of sufficient conditions for strategyproofness.
\item We use an analytic approach to \emph{derive} strategyproof mechanisms which are more competitive than all prior strategyproof mechanisms.
\item We improve the linear-program-based proof of Guo and Conitzer~\cite{GuoC2010} to show new upper bounds on competitive ratios.
\item We provide the first \emph{compact} proof on upper bound of competitiveness.
\end{itemize}

For the cases with any number of items, we build on the Partial Allocation mechanisms introduced by Cole et al.~\cite{ColeGG2013-EC,ColeGG2013-AAMAS}
to design a strategyproof mechanism which is $0.67776$-competitive, breaking the $\frac 23$ barrier.

We also propose a new sub-class of strategyproof mechanisms for any numbers of agents and items,
which we call it \emph{Dynamic-Increasing-Price} mechanisms,
where each agent purchases the items using virtual money, and the prices of the items depend on other agents' preferences.
\end{abstract}

\section{Introduction}

Competition for resources is one of the most primitive activities of human.
The problems of distributing resources among competing agents are found wherever human exists.
We have developed various systems and tools, e.g., economic markets, auctions/mechanisms, voting systems and money,
to facilitate the distributions and exchanges of resources.
The study of these problems in various settings is now a popular topic concerning both computer scientists and economists,
in the field coined as \emph{Multiagent Resource Allocation};
see \cite{Chev2006} for a fairly recent survey.

While money is provably reducing the complexity of many such problems, in many scenarios monetary transfer is inapplicable;
canonical examples include hospital/facility location determination, peer-to-peer sharing network,
and distribution of computing resources among staffs in the same company.
Mechanism design problems without monetary transfer are thus naturally motivated.

Briefly speaking, a mechanism is a combination of a communication protocol and an algorithm,
for agents to reveal their preferences to an auctioneer, and for the auctioneer to determine a \emph{good allocation} of resources.
There are different ways to interpret the meaning of ``good allocation'', including
social welfare, various measures of fairness (often coined as ``cake cutting''),
or revenue of the auctioneer (when monetary transfer is feasible).
Independent of the interpretation, a favourable feature of a mechanism is strategyproofness
--- the mechanism \emph{motivates} the agents to reveal their preferences \emph{truthfully}.
Strategyproof (SP) mechanism design without monetary transfer has been studied under various contexts, e.g.,\cite{SchummerV2002,ProcacciaT2009,DokowFMN2012,DughmiG2010}.

We focus on the problem formulated by Guo and Conitzer~\cite{GuoC2010}:
design SP mechanisms for allocating divisible items among two agents who have linear preferences over the items.
The target is to attain social welfares which are competitive against those in the first-best mechanism (which is not SP).

In contrast to most of the prior work, we take on an analytic approach for the problem.
While all known SP mechanisms (which we will discuss next) are somewhat naturally motivated,
they do not shed any insight on how an optimal mechanism should look like.
We will present results obtained using analytic methods,
which will suggest that analytical insights are necessary for seeking the optimal mechanisms.

\subsection{Related Work}

Guo and Conitzer~\cite{GuoC2010} considered a sub-class of SP mechanisms called Swap-Dictatorial mechanisms,
in which each agent is a \emph{dictator} with probability $\frac 12$,
who chooses her favourite allocation from a predefined set of allowable allocations,
and the other agent is allocated the remaining items.\footnote{Swap-Dictatorial mechanisms can be generalized to more than two agents,
by first generating a random order of the agents, and then each agent takes turn to choose her favourite allocation.}
They studied two sub-classes of Swap-Dictatorial mechanisms, Increasing-Price (IP) mechanisms and Linear-Increasing-Price (LIP) mechanisms.
For the case with two items, they showed that there is a LIP mechanism which is $0.828$-competitive against the first-best mechanism;
they used a linear program to show that no SP mechanism can be better than $0.841$-competitive.
They also showed that as the number of items goes to infinity, IP and LIP mechanisms have maximal competitiveness of $\frac 12$.

Han et al.~\cite{HanSTZ2011} showed a number of upper bound results on the competitiveness of SP mechanisms, when the numbers of agents and/or items increase.
In particular, they showed that no swap-dictatorial mechanism can be better than $\left(\frac 12 + o_m(1)\right)$-competitive for $m$ items.
In addition, they proved the following characterization result:
in the case with two items, if a mechanism $A$ is symmetric and second order continuously differentiable,
then $A$ is SP if and only if $A$ is swap-dictatorial.

Cole, Gkatzelis and Goel \cite{ColeGG2013-EC} proposed another sub-class of SP mechanisms (for any number of agents)
called \emph{Partial Allocation} (PA) mechanisms, which are not swap-dictatorial.
They showed in another work \cite{ColeGG2013-AAMAS} that a variant of PA mechanism is $\frac 23$-competitive for two agents and any number of items.

Recently, non-SP resource allocation mechanisms are also under study.
The mechanisms are considered as games and the bids of agents are strategies.
The canonical measure of efficiency is the Price of Anarchy.
See~\cite{FLZ2009,BranzeiCDFF2014} and the references therein for more details.

We note that in every swap-dictatorial mechanism, all items are always completely allocated among the agents
 --- the auctioneer never holds some of the items from being allocated. However, this is not true in PA mechanism.
We say a mechanism is \emph{full} if all items are always completely allocated among the agents,
and say it is \emph{partial} otherwise.

\subsection{Our Contribution}

Our main contribution is to use analytic methods to \emph{derive} SP mechanisms with competitiveness better than those previously known.
We also improve a linear-program (LP) based proof to show new upper bounds on the competitive ratios for SP full and partial mechanisms.
In addition, we provide the first \emph{compact} upper bound proof.

In Sections \ref{sect:suff-condition}---\ref{sect:human-upper-bound}, we focus on the case with two agents and two items.
In Section \ref{sect:suff-condition}, we first prove a characterization of symmetric SP mechanisms,
which is essentially the same as the Rochet's characterization~\cite{Rochet1985}.
Then we provide a set of sufficient conditions for symmetric SP mechanisms.
We note that while our set of sufficient conditions and the characterization in~\cite{HanSTZ2011} are both of analytical flavor,
the two results are not comparable:
their result focuses on conditions that yield equivalence between SP mechanisms and swap-dictatorial mechanisms (which must be full mechanisms),
while our result is applicable for a broad sub-class of partial mechanisms,
and also mechanisms which are not second order continuously differentiable.

In Section \ref{sect:fully-two-items}, we look into the solution to a LP of Guo and Contizer~\cite{GuoC2010}
for making a few observations and heuristic assumptions, which allow us to \emph{derive} a $\frac 56$-competitive full mechanism;
we believe it is an optimal full mechanism.\footnote{There may be more than one optimal full mechanisms.}
In Section \ref{sect:partial-two-items}, by using our set of sufficient conditions, we consider a sub-family of SP partial mechanisms,
and show that one of such mechanisms is strictly better than $\frac 56$-competitive.
This may be surprising to some practitioners, since it suggests that in general,
the competitiveness can be improved by suitably holding a fraction of items from being allocated.

Guo and Conitzer used the LP to show an upper bound of $0.841$ on the competitiveness of SP full mechanisms.
We will discuss how to prune out a lot of \emph{unnecessary} constraints from their LP.
This allows us to solve the LP with much refined resolution,
and to improve the upper bound to $\frac 56 + \epsilon$, where $\epsilon < 10^{-9}$; we believe the final answer is $\frac 56$.
With a minor modification to their LP, we show an upper bound of $0.8644$ on the competitiveness of SP partial mechanisms.

While the LP-based upper bound proofs are \emph{legitimate},
they may look unsatisfactory to some researchers, due to two reasons.
First, such proofs are hardly verifiable by researchers without the use of a computer.
Second, such LPs are extremely huge for three or more items and thus not solvable in practice,
so they shed no insight for providing a better upper bound when the number of items increases.
Therefore, a \emph{compact} upper bound proof, i.e., a proof which can be easily verifiable by researchers, is preferred.
In Section \ref{sect:human-upper-bound}, we use the SP characterization of Rochet~\cite{Rochet1985}
to provide the first compact proof; the upper bound is $0.9523$.
While this upper bound is worse than those yielded by LP-based proofs,
the compact proof is worth an attention since it might shed insight for generalizations.

In Section \ref{sect:avg-PA}, we consider the cases with two agents and any number of items.
By taking a suitable \emph{average} of some PA mechanisms of Cole et al.~\cite{ColeGG2013-EC,ColeGG2013-AAMAS},
we design a SP mechanism which is at least $0.67776$-competitive.

In Section \ref{sect:DIP}, we propose a new sub-class of SP mechanisms for any number of agents and items,
called \emph{Dynamic-Increasing-Price} (DIP) mechanisms.
DIP mechanism is similar to IP mechanism in the sense that both introduce \emph{virtual money} and \emph{virtual prices}.
However, there is no dictator-swapping process in DIP, and DIP is not swap-dictatorial in general.
Also, the prices in an IP mechanism is independent of agents' preferences, but in a DIP mechanism,
for each agent, the prices of the items depend on other agents' preferences.
In other words, a DIP mechanism enforces all agents to be \emph{complete price takers}.
We note that DIP is well motivated by the classical context of markets: when the scale of the market is large,
each agent in the market has tiny effect on the prices,
so the prices she face are \emph{almost} completely depending on other agents preferences.
We show that the $\frac 56$-competitive mechanism is a DIP mechanism.

\section{Preliminaries}

\paragraph{Problem Setting.}
We study the problem of allocating $m\geq 2$ divisible items, each of one unit, among two agents, referred to as agents 1 and 2.
A vector $(c_1,\cdots,c_m)$ is \emph{normalized} if each $c_j\geq 0$ and $\sumjm c_j = 1$.
Each agent $i$ has a normalized\footnote{See~\cite[Section 2]{GuoC2010} for an explanation on why the utility functions are normalized.}
linear utility function $u_i(\vec{x}_i) = \sumjm \uij \xij$, where $\sumjm \uij = 1$;
her utility function is identified to the normalized utility vector $\vecu_i :=(u_{i1},\cdots,u_{im})$.
Each agent $i$ reports to mechanism $A$ a normalized bid vector $\vecb_i = (b_{i1},\cdots,b_{im})$.
The mechanism $A$, based on the bids, allocates $A_{ij} (\vecb_1,\vecb_2)$ unit of item $j$ to agent $i$.
The allocation must be \emph{feasible}, i.e., for any $i,j$, $A_{ij} (\vecb_1,\vecb_2)\geq 0$ and
for any $j$, $A_{1j} (\vecb_1,\vecb_2) + A_{2j} (\vecb_1,\vecb_2) \leq 1$.

Let
$$u_i^A(\vecb_1,\vecb_2) ~:=~ \sumjm \uij\cdot A_{ij}(\vecb_1,\vecb_2),$$
which is the utility attained by agent $i$ when agent 1 bids $\vecb_1$ and agent 2 bids $\vecb_2$.
Mechanism $A$ is \emph{strategyproof} (SP) if for any agent $i$ and for any normalized vectors $\vecb_i,\vecb_{3-i}$,
$$u_i^A(\vecu_i,\vecb_{3-i}) ~\geq~ u_i^A(\vecb_i,\vecb_{3-i}),$$
i.e., agent $i$ is always better off to bid her true utility vector.
Strategyproofness is generally accepted as a favourable feature of a mechanism,
since it discourages agents from having strategic consideration for reporting bids.

Let
$$\swa(\vecu_1,\vecu_2) ~:=~ u_1^A(\vecu_1,\vecu_2) + u_2^A(\vecu_1,\vecu_2),$$
which is the social welfare when both agents bid truthfully to mechanism $A$.
Let
$$\swopt ~:=~ \sumjm ~\max\{u_{1j},u_{2j}\},$$
which is the maximum possible social welfare among all feasible allocations;
the feasible allocation that attains the maximum possible social welfare is called the \emph{first-best allocation}.
A SP mechanism $A$ is $\alpha$-competitive if for any $\vecu_1,\vecu_2$,
$\swa(\vecu_1,\vecu_2) \geq \alpha\cdot \swopt$.
Our target is to design SP mechanisms with high competitive ratios.

\paragraph{Useful Definitions, Facts and Tool.}
%
It is known that every $\alpha$-competitive SP mechanism has a corresponding
symmetric-over-agents and symmetric-over-items $\alpha$-competitive SP mechanism \cite[Claim 1]{GuoC2010}.
Thus, from now on we focus only on such symmetric SP mechanisms.

When the agents' utility functions are linear, any \emph{weighted average} over SP mechanisms is also SP:
if $A^1,\cdots,A^k$ are SP mechanisms, then $\bar{A}$, defined by the allocation rule
$$\bar{A}_{ij}(\vecb_1,\vecb_2) ~:=~ \sum_{\ell=1}^k \beta_\ell \cdot A^\ell_{ij}(\vecb_1,\vecb_2),$$
where the $\beta_\ell$'s are positive and $\sum_{\ell=1}^k \beta_\ell = 1$, is also SP.
We will write $\bar{A}$ as $\sum_{\ell=1}^k \beta_\ell \cdot A^\ell$.

For any utility functions of the agents $u_1,u_2$, let their \emph{attainable utility region} (AUR) be
$$\aur(u_1,u_2) ~:=~\left\{(r_1,r_2)\left|~\exists\text{ a feasible allocation }(\vecx_1,\vecx_2)\text{~s.t.~}
u_1(\vecx_1) = r_1\text{ and }u_2(\vecx_2) = r_2\right.\right\}.$$

\begin{prop}\label{prop:aur-convex}
If $u_1,u_2$ are increasing, concave and continuous functions, then \emph{$\aur(u_1,u_2)$} is a convex subset of $\rr^2$.
\end{prop}

\begin{proof}
Suppose $(r_1,r_2),(r'_1,r'_2)\in\aur(u_1,u_2)$ and $\beta$ is any real number in $[0,1]$.
It suffices to prove that $(\beta r_1 + (1-\beta) r'_1, \beta r_2 + (1-\beta) r'_2)\in \aur(u_1,u_2)$.

Let $(\vecx_1,\vecx_2)$ be a feasible allocation such that for $i=1,2$, $u_i(\vecx_i) = r_i$.
Let $(\vecx_1',\vecx_2')$ be a feasible allocation such that for $i=1,2$, $u_i(\vecx_i') = r_i'$.
Note that $(\beta\cdot\vecx_1 + (1-\beta)\cdot\vecx_1',\beta\cdot\vecx_2 + (1-\beta)\cdot\vecx_2')$ is also a feasible allocation.
Since $u_i$ is a concave function,
$$u_i(\beta\cdot\vecx_i + (1-\beta)\cdot\vecx_i') ~\geq~ \beta \cdot u_i(\vecx_i) + (1-\beta)\cdot u_i(\vecx_i') ~=~ \beta r_i + (1-\beta) r_i'.$$
Then, since $u_i$ is increasing and continuous, there exists $\gamma_i\in [0,1]$ such that
$$u_i(\beta\gamma_i\cdot\vecx_i + (1-\beta)\gamma_i\cdot\vecx_i') ~=~ \beta r_i + (1-\beta) r_i',$$
i.e., $(\beta\gamma_1\cdot\vecx_1 + (1-\beta)\gamma_1\cdot\vecx_1',\beta\gamma_2\cdot\vecx_2 + (1-\beta)\gamma_2\cdot\vecx_2')$ is a feasible allocation
that verifies $(\beta r_1 + (1-\beta) r'_1, \beta r_2 + (1-\beta) r'_2)\in \aur(u_1,u_2)$.
\end{proof}

\section{A Set of Sufficient Conditions for Strategyproofness}\label{sect:suff-condition}

In Sections \ref{sect:suff-condition}---\ref{sect:human-upper-bound}, we focus on the case with two items.
In this case, each normalized utility vector has the form $(t,1-t)$, which is essentially single-parameter.
We assume that each agent $i$ bids a number $b_i\in [0,1]$,
which is supposed to be the first entry of her normalized utility vector.
A symmetric mechanism $A$ can be described by a single function $A:[0,1]^2\ra\rr^+$, such that
\begin{align*}
A_{11}(b_1,b_2) \equiv A(b_1,b_2) ~~&~~A_{12}(b_1,b_2) \equiv A(1-b_1,1-b_2)\\
A_{21}(b_1,b_2) \equiv A(b_2,b_1) ~~&~~A_{22}(b_1,b_2) \equiv A(1-b_2,1-b_1)
\end{align*}

\medskip

In this section, we first prove a characterization of SP symmetric mechanisms (Theorem \ref{thm:Rochet}),
which follows almost directly from a characterization result of Rochet~\cite[Theorem 1]{Rochet1985}; we will provide a self-contained proof.
Then we use the characterization to provide a set of sufficient conditions for strategyproofness (Theorem \ref{thm:SP-char}),
which will be used in the next three sections.

Let
$$\ucap^A(b_1,b_2) := b_1 \cdot A(b_1,b_2) + (1-b_1) \cdot A(1-b_1,1-b_2),$$
which is the utility attained by agent 1 if her true utility vector is $(b_1,1-b_1)$, she bids truthfully, and agent 2 bids $b_2$.

\begin{theorem}[{\cite[Theorem 1]{Rochet1985}}] \label{thm:Rochet}
Let $A$ be a symmetric mechanism for two items. $A$ is SP if and only if
\begin{enumerate}
\item[(a)] for any fixed $b_2\in [0,1]$, $\ucap^A(b_1,b_2)$ is a convex function of $b_1$, and
\item[(b)] for any fixed $b_2\in [0,1]$, $z := A(t_1,b_2) - A(1-t_1,1-b_2)$ is a sub-gradient of $\ucap^A(b_1,b_2)$ at $b_1 = t_1$,
i.e., for any $b_1\in [0,1]$,
$$\ucap^A(b_1,b_2) ~\geq~ \ucap^A(t_1,b_2) + z \cdot (b_1 - t_1).$$
\end{enumerate}
\end{theorem}

We note that since we are considering symmetric mechanisms, in the above theorem,
stating the conditions (a) and (b) w.r.t.~agent $1$ only is without loss of generality.

\begin{proof}
If $A$ is SP, then $\ucap^A(b_1,b_2) = \sup_{b_1'} b_1 \cdot A(b_1',b_2) + (1-b_1) \cdot A(1-b_1',1-b_2)$,
which is a supremum of linear functions of $b_1$, thus $\ucap^A(b_1,b_2)$ is convex w.r.t.~$b_1$, i.e., condition (a) holds.

Next, we show that condition (b) is equivalent to strategyproofness as below; the second, third and the fourth statements below are equivalent
since the R.H.S.~of the inequalities in them are indeed identical.
\begin{align*}
 ~&~A\text{ is SP}\\
\Leftrightarrow~&~\forall b_1,b_2,t_1,~\ucap^A(b_1,b_2) ~\geq~ b_1\cdot A(t_1,b_2) + (1-b_1)\cdot A(1-t_1,1-b_2)\\
\Leftrightarrow~&~\forall b_1,b_2,t_1,~\ucap^A(b_1,b_2) ~\geq~ t_1\cdot A(t_1,b_2) + (1-t_1)\cdot A(1-t_1,1-b_2) + z\cdot (b_1 - t_1)\\
\Leftrightarrow~&~\forall b_1,b_2,t_1,~\ucap^A(b_1,b_2) ~\geq~ \ucap^A(t_1,b_2) + z\cdot (b_1 - t_1)\\
\Leftrightarrow~&~\text{condition (b) holds.}
\end{align*}
\end{proof}

\begin{theorem}\label{thm:SP-char}
Let $A$ be a symmetric mechanism, described by function $A(b_1,b_2)$.
If for any fixed $b_2$, $A(b_1,b_2)$ is increasing, continuous and piecewise continuously differentiable w.r.t.~$b_1$,
and if for any $t_1,t_2\in [0,1]$, the equality\footnote{A clarification on the perhaps misleading notation:
at any specific point $(y_1,y_2)$, $\frac{\partial A}{\partial b_1} (y_1,y_2)$ is
the value of the partial derivative of $A$ w.r.t.~its first parameter at that point.
To be crystal clear, $\frac{\partial A}{\partial b_1} (y_1,y_2) = \lim_{\delta\rightarrow 0} \frac{A(y_1+\delta,y_2) - A(y_1,y_2)}{\delta}$.}
\begin{equation}\label{eq:SP-char}
t_1 \cdot \frac{\partial A}{\partial b_1} (t_1,t_2) = (1-t_1) \cdot \frac{\partial A}{\partial b_1} (1-t_1,1-t_2)
\end{equation}
holds within each piecewise interval, then $A$ is SP.
\end{theorem}

\begin{proof}
Recall that
$$\ucap^A(b_1,b_2) ~=~ b_1\cdot A(b_1,b_2) + (1-b_1)\cdot A(1-b_1,1-b_2).$$
By the assumptions on $A$, $\frac{\partial \ucap^A}{\partial b_1}(t_1,t_2)$ exists everywhere (except perhaps at the endpoints of the piecewise intervals),
and its value is
$$A(t_1,t_2) ~-~ A(1-t_1,1-t_2) ~+~ t_1 \cdot \frac{\partial A}{\partial b_1} (t_1,t_2) ~-~ (1-t_1)\cdot \frac{\partial A}{\partial b_1} (1-t_1,1-t_2).$$
The final two terms cancel out due to \eqref{eq:SP-char}. Hence,
\begin{equation}\label{eq:SP-char-Rochet}
\frac{\partial \ucap^A}{\partial b_1}(t_1,t_2) ~=~ A(t_1,t_2) ~-~ A(1-t_1,1-t_2).
\end{equation}
Since $A$ is continuous and piecewise continuously differentiable w.r.t.~its first parameter,
at any endpoint $(t_1,t_2)$ of a piecewise interval,
the left and right partial derivatives of $\ucap^A$ w.r.t.~its first parameter are equal,
i.e., $\frac{\partial \ucap^A}{\partial b_1}(t_1,t_2)$ exists at the endpoint too.

When $t_1$ increases, $A(t_1,t_2)$ increases but $A(1-t_1,1-t_2)$ decreases.
By \eqref{eq:SP-char-Rochet}, $\frac{\partial \ucap^A}{\partial b_1}(t_1,t_2)$ increases with $t_1$ within each piecewise interval.
Thus, $\ucap^A(b_1,b_2)$ is convex w.r.t.~$b_1$, i.e., condition (a) in Theorem \ref{thm:Rochet} holds.

\eqref{eq:SP-char-Rochet} and condition (a) imply that $A(t_1,t_2) - A(1-t_1,1-t_2)$ is a subgradient of $\ucap^A(b_1,t_2)$ at $b_1 = t_1$,
i.e., condition (b) in Theorem \ref{thm:Rochet} holds. Then by Theorem \ref{thm:Rochet}, $A$ is SP.
\end{proof}

Note that the allocation functions of Partial Allocation mechanisms in~\cite{ColeGG2013-EC,ColeGG2013-AAMAS} are discontinuous,
so Theorem \ref{thm:SP-char} is not applicable.

\section{A $\frac 56$-Competitive Full Mechanism for Two Items}\label{sect:fully-two-items}

Guo and Conitzer~\cite{GuoC2010} introduced the linear program (LP) below,
which represents the optimal full mechanism when the bids are restricted to be multiples of $1/N$ for some integer $N$.
Let $[N]$ denote the set of all multiples of $1/N$ which are between zero and one.
\begin{align}
&~~~~~~\max \lambda\nonumber\\
&\forall t_1,t_1',t_2\in [N],&&\ucap^A(t_1,t_2) \geq t_1 \cdot A(t_1',t_2) + (1-t_1)\cdot A(1-t_1',1-t_2); & \text{(stragyproofness)}\nonumber\\
&\forall t_1,t_2\in [N],&& \swa(t_1,t_2)\geq (1+|t_1-t_2|)\lambda; & \text{(competitiveness)}\nonumber\\
&\forall t_1,t_2\in [N],&& A(t_1,t_2) + A(t_2,t_1) = 1; & \label{eq:fully}\\
&\forall t_1,t_2\in [N],&& A(t_1,t_2)\geq 0. & \nonumber
\end{align}
They solved the LP with $N=50$. The optimal $\lambda$ value, which is $0.841$,
is an upper bound on the optimal competitiveness of SP full mechanisms for two items.

The LP has $\Theta(N^2)$ variables and $\Theta(N^3)$ constraints, which is efficiently solvable only for small $N$.
However, one would expect that the strategyproofness constraints with large $|t_1-t_1'|$ are unnecessary.
So we keep only those constraints with $|t_1-t_1'| = 1/N$.
This reduces the number of constraints to $\Theta(N^2)$,
allowing us to solve the LP with a much refined resolution of $N=400$.
We obtain an improved upper bound of $\frac 56 + \epsilon$, where $\epsilon < 10^{-9}$.
We believe that $\frac 56$ is the final answer.

We make two observations from the solution to the LP, and make two heuristic assumptions.
We then use the observations and assumptions to \emph{derive} a full mechanism which is $\frac 56$-competitive.

\medskip

\noindent\textbf{Observation 1. }There exists a function $f:[0,1]\ra\rr$ such that $A(t_1,t_2) = f(t_1) - f(t_2) + \frac 12$.
Furthermore, $f$ is increasing, continuous and piecewise differentiable.

\medskip

\noindent\textbf{Observation 2. }For all $t\in \left[0,\frac 15\right]$, $f(t) = 0$. For all $t\in\left[\frac 45,1\right]$, $f(t)$ is a constant.

\medskip

\noindent\textbf{Assumption 3. }The function $A(t_1,t_2)$ satisfies the equality \eqref{eq:SP-char},
except at points where $t_1\in \left\{1/5,4/5\right\}$.

\medskip

With Observation 1, Observation 2 and Assumption 3, \eqref{eq:SP-char} yields
\begin{equation}\label{eq:SP-char-f}
\forall t\in [0,1]\setminus\left\{1/5,4/5\right\},~~t f'(t) = (1-t) f'(1-t).
\end{equation}

%

With Observation 2, the social welfare attained when $t_2=0$ is $1+t_1 \left(f(t_1) - f(1-t_1) + \frac 12\right)$,
while $\text{\textsf{SW}}_{\text{\textsf{OPT}}}(t_1,t_2) = 1+t_1$.
For the mechanism to be $\frac 56$-competitive, the following inequality must hold:
$1+t_1\left(f(t_1) - f(1-t_1) + \frac 12\right) \geq \frac 56 (1+t_1)$, or equivalently
\begin{equation}\label{eq:f-diff}
f(t) - f(1-t) \geq \frac 13 - \frac{1}{6t}.
\end{equation}
Observe that by \eqref{eq:SP-char-f}, once the values of $f(t)$ for $0\leq t\leq \frac 12$ are known,
the values of $f(t)$ for $\frac 12 \leq t\leq 1$ can be determined.
We now state the final heuristic assumption: in \eqref{eq:f-diff}, the equality holds for $t\in \left[\frac 15,\frac 12\right]$.

\medskip

\noindent\textbf{Assumption 4. }$\forall t\in\left[\frac 15, \frac 12\right]$, $f(t) - f(1-t) = \frac 13 - \frac{1}{6t}$.

\medskip

With \eqref{eq:f-diff} and Assumption 4, we can solve $f$ using calculus, which is:
\begin{equation}\label{eq:f-explicit}
f(t) = \begin{cases}
0, & t\in\left[0,\frac 15\right];\\
\frac 56 - \frac{1}{6t} - \frac 16 \ln(5t), & t\in\left[\frac 15,\frac 12\right];\\
\frac 12 - \frac 16 \ln(5-5t), & t\in\left[\frac 12,\frac 45\right];\\
\frac 12, & t\in \left[\frac 45,1\right].
\end{cases}
\end{equation}

\begin{theorem}
The full mechanism $A$ as described in Assumption 1 and \eqref{eq:f-explicit} is feasible, SP and $\frac 56$-competitive.
\end{theorem}

\begin{proof}
Feasibility trivially holds. Strategyproofness follows from Observation 1, Assumption 3 and Theorem \ref{thm:SP-char}.

For competitiveness, note that
$$\swa(t_1,t_2) ~=~ 1 + (t_1 - t_2) \left[ f(t_1) - f(t_2) - f(1-t_1) + f(1-t_2) \right].$$
Showing that $A$ is $\frac 56$-competitive is equivalent to showing that $\swa(t_1,t_2) \geq \frac 56 (1 + |t_1-t_2|)$,
which can be done with an appropriate case analysis (which is needed due to the piecewise definition of $f$) and simple calculus;
we skip the details.
\end{proof}

\section{A Partial Mechanism for Two Items -- Strictly Better than $\frac 56$-Competitive}\label{sect:partial-two-items}

By changing the equality sign in \eqref{eq:fully} to a $\leq$ sign, Guo and Conitzer's LP covers partial mechanisms also.
The modified LP provides an upper bound of $0.8644$.
We look into its solution, as we did in Section \ref{sect:fully-two-items}, but we do not recognize a nice pattern.

Since we solve the modified LP with high resolution, we \emph{believe} that an optimal partial mechanism attains competitive ratio close to $0.8644$,
which beats the $\frac 56+\epsilon$ upper bound for full mechanism.
Yet, to formally \emph{prove} that an optimal partial mechanism is strictly better than an optimal full mechanism,
we ought to provide a concrete, and preferably compact, SP partial mechanism which is strictly better than $(\frac 56+\epsilon)$-competitive.
This is the purpose of the current section.

Let $f_1 : \left[0,1/2\right]\ra \rr$ and $f_2: \left[1/2, 1\right] \ra \rr$ be two increasing and continuously differentiable functions
such that for all $t\in \left[0,1/2\right]$, $t f_1'(t) = (1-t) f_2'(1-t)$, and $f_1(0)=f_2\left(1/2\right) = 0$.
Also, let $Q,R: [0,1]\ra\rr^+$ be two functions.
Then define the function
$$A(t_1,t_2) ~:=~
\begin{cases}
Q(t_2)\cdot f_1(t_1) + R(t_2),&t_1\in\left[0,1/2\right];\\
A\left(\frac 12,t_2\right) + Q(1-t_2)\cdot f_2(t_1),&t_1\in\left(1/2,1\right].
\end{cases}$$
It is easy to verify that the above function $A$ satisfies all conditions required in Theorem \ref{thm:SP-char},
and thus it yields a SP mechanism, modulo feasibility constraint.

Our strategy is to pick some choice of $f_1,f_2$, and then use an LP to find out $Q,R$,
such that $A$ is feasible and attains a good competitiveness.
As before, we formulate the LP with bids restricted to be multiples of $1/N$ for some integer $N$.
The LP is stated below; note that we impose a slightly stricter feasibility constraint,
in which only $(1-\delta)$ fraction, for some $\delta>0$, of each item can be allocated.
The reason will be clear later.
\begin{align*}
&~~~~~~\max \lambda\\
&\forall t_1,t_2\in [N],&&~ A(t_1,t_2) + A(t_2,t_1)\leq 1-\delta;\\
&\forall t_1,t_2\in [N],&&~ \swa(t_1,t_2)\geq (1+|t_1-t_2|)\lambda.\\
&\forall t\in [N],&&~ Q(t),R(t)\geq 0.
\end{align*}
It is easy to verify that the above program is an LP with variables $Q(t),R(t)$ for $t\in [N]$, plus an extra variable $\lambda$.

Lacking further insight on how a good choice of $f_1,f_2$ should be, we try the natural candidate $f_1(t) := t$,
and hence $f_2(t) := \ln (2t) - t + 1/2$. Then we solve the above LP with resolution $N=1000$ and $\delta = 2.92/2000$.
The optimal $\lambda$ is larger than $0.835524$. Let the optimal solution be $\Qs,\Rs$.
We note that the maximum entry in $\Qs$ is less than $1.46$.

Now, we are ready to describe the desired symmetric SP mechanism $\tA$.
$\tA$ takes $t_1,t_2$ as bids from the two agents.
Let $\tt_1,\tt_2$ be the values by rounding $t_1,t_2$ to its nearest multiple of $1/N$.
$$\tA(t_1,t_2) ~:=~
\begin{cases}
\Qs(\tt_2)\cdot f_1(t_1) + \Rs(\tt_2),&t_1\in\left[0,1/2\right];\\
\tA\left(\frac 12,\tt_2\right) + \Qs(1-\tt_2)\cdot f_2(t_1),&t_1\in\left(1/2,1\right].
\end{cases}$$

%

The rounding is needed because the domains of $\Qs,\Rs$ are $[N]$.
The rounding does not destroy strategyproofness; one can verify that $\tA$ is SP using Theorem \ref{thm:SP-char}.
\footnote{In general, if an allocation function for agent $i$, denoted by $A_i'(\vecu_i,\vecu_{-i})$, yields a SP mechanism,
then any other allocation function $A_i''(\vecu_i,\vecu_{-i}) \equiv A_i'(\vecu_i,T(\vecu_{-i}))$,
where $T$ is an \emph{arbitrary} function with range compatible with the domain of the second parameter of $A_i'$, yields a SP mechanism too.}


Since the maximum entry of $\Qs$ is less than $1.46$, and since the derivatives of $f_1,f_2$ are bounded by $1$,
$\frac{\partial \tA}{\partial b_1} (b_1,b_2) < 1.46$ for all $b_1,b_2$.
Also, note that $|t_i - \tt_i| \leq 1/2000$. Due to the first constraint of the LP above, $\tA(\tt_1,\tt_2) + \tA(\tt_2,\tt_1) \leq 1-\delta$.
Thus,
$$\tA(t_1,t_2) + \tA(t_2,t_1) ~<~ \tA(\tt_1,\tt_2) + \frac{1.46}{2000} + \tA(\tt_2,\tt_1) + \frac{1.46}{2000} ~\leq~ 1.$$
This verifies the feasibility of $\tA$.

To bound the competitiveness of $\tA$, first note that by Theorem \ref{thm:Rochet},
$$\ucap^{\tA}(t_1,\tt_2) \geq \ucap^{\tA}(\tt_1,\tt_2) - \frac{1}{2000}~~~~\text{and}~~~~\ucap^{\tA}(t_2,\tt_1) \geq \ucap^{\tA}(\tt_2,\tt_1) - \frac{1}{2000}.$$
From the LP, we have
$$\frac{\ucap^{\tA}(\tt_1,\tt_2) + \ucap^{\tA}(\tt_2,\tt_1)}{1+|\tt_1-\tt_2|} ~\geq~ 0.835524.$$
With the above two sets of inequalities, we proceed a simple error analysis to show that the competitiveness of $\tA$ with bids $t_1,t_2$ is
$$\frac{\ucap^{\tA}(t_1,t_2) + \ucap^{\tA}(t_2,t_1)}{1+|t_1-t_2|} ~=~ \frac{\ucap^{\tA}(t_1,\tt_2) + \ucap^{\tA}(t_2,\tt_1)}{1+|t_1-t_2|} ~\geq~ 0.833689,$$
which is strictly larger than $\frac 56+\epsilon$.

\medskip

To program mechanism $\tA$, we need to store the values of $\Qs,\Rs$ at $1001$ discrete values,
and to compute $f_1,f_2$ at arbitrary real values in their domain.
While this may not look compact to some people, $\tA$ is a concrete SP partial mechanism that breaks the $\frac 56$ barrier.

\section{A Compact Upper Bound Proof}\label{sect:human-upper-bound}

In this section, we provide a compact upper bound proof on the competitiveness of SP mechanisms.

First of all, we make the following qualitative observation.
Suppose $\vecu_i = (t_i,1-t_i)$.
For some $t_1,t_2$, if agent $1$ earns a \emph{too} high utility from the mechanism,
then the utility earned by agent $2$ has to be very low, forcing a low competitive ratio;
conversely, if agent $1$ earns \emph{too} low a utility, this forces a low competitive ratio too.
Thus, to attain a high competitive ratio $h$, there is a restricted range of utility values that each agent can earn.
Geometrically, the utilities earned by the agents must lie in the intersection of
$\aur(t_1,t_2)$ and $R_h(t_1,t_2) := \{(r_1,r_2)\,|\,r_1+r_2 \geq h \cdot \swoptnull(t_1,t_2)\}$.

Briefly, our proof strategy is: for some $h,t_1,t_2$,
since the utilities earned by both agents are restricted to certain range, the allocations are restricted too.
Then due to Theorem \ref{thm:Rochet}, the subgradients of $\ucap^A$ at certain points are also restricted.
We then show that if $h$ is too high, the above restrictions add up to forbid the existence of $\ucap^A$.

Let
\begin{align*}
U_h(t_1,t_2) &:= \max\{r_1\,|\,\exists (r_1,r_2)\in\aur(t_1,t_2)\cap R_h(t_1,t_2)\}\\
L_h(t_1,t_2) &:= \min\{r_1\,|\,\exists (r_1,r_2)\in\aur(t_1,t_2)\cap R_h(t_1,t_2)\}.
\end{align*}

\begin{theorem}
For the case with two items, no SP mechanism is better than $0.9523$-competitive.
\end{theorem}

\begin{proof}
By~\cite[Claim 1]{GuoC2010}, it suffices to prove that no symmetric SP mechanism is better than $0.9523$-competitive.

Suppose the contrary that there exists an $h$-competitive symmetric SP mechanism $A$, where $h=0.9523$.
We compute the following functions explicitly:
$$
U_h(t_1,0.1) = \begin{cases}
1 - \frac{(\frac{11}{10}-t_1)h-1}{\frac{1}{10}-t_1} t_1, & t_1\in \left[0,\frac{11}{10}-\frac{1}{h}\right];\\
1, & t_1 \in \left[\frac{11}{10}-\frac{1}{h},\frac{1}{h}-\frac{9}{10}\right];\\
1 + \frac{(t_1+\frac{9}{10})h-1}{t_1-\frac{1}{10}} (t_1-1), & t_1 \in \left[\frac{1}{h}-\frac{9}{10},1\right].
\end{cases}
$$

$$L_h(t_1,0.1) = \begin{cases}
\frac{(\frac{11}{10}-t_1)h-1}{\frac{1}{10}-t_1} (1-t_1), & t_1\in \left[0,\frac{11}{10}-\frac{1}{h}\right];\\
0, & t_1 \in \left[\frac{11}{10}-\frac{1}{h},\frac{1}{h}-\frac{9}{10}\right];\\
\frac{(t_1+\frac{9}{10})h-1}{t_1-\frac{1}{10}} t_1, & t_1 \in \left[\frac{1}{h}-\frac{9}{10},1\right].
\end{cases}
$$

$$U_h(t_1,0) = \min\left\{1,1 + \frac{(t_1+1)h-1}{t_1} (t_1-1)\right\}.$$
In particular, we have
$$0.4753 = L_h(0,0.1) \leq \ucap^A(0,0.1) \leq U_h(0,0.1) = 1.$$

Let $q := \ucap^A(0,0.1)$. Then $A(1,0.9) = q$. By feasibility, $A(0,0.1)\geq 0$.
Thus, $A(0,0.1) - A(1,0.9) \geq -q$.
By Theorem \ref{thm:Rochet}, for any $t_1\in [0,1]$,
\begin{equation}\label{eq:support1}
\ucap^A(t_1,0.1) \geq q - q t_1.
\end{equation}

Since $A$ is $h$-competitive, $\ucap^A(0.1,0) \geq \frac{11}{10}h-q$.
Since $A(1,0.9) = q$, $A(0.9,1) \leq 1-q$.
Then we have
$$A(0.1,0) ~=~ \frac{\ucap^A(0.1,0)-0.9\cdot A(0.9,1)}{0.1} ~\geq~ 11h-9-q,$$
and hence $A(0.1,0) - A(0.9,1) \geq 11h-10$.
By Theorem \ref{thm:Rochet}, for any $t_1\in [0.1,1]$,
\begin{equation}\label{eq:support2}
\ucap^A(t_1,0) \geq \frac{11}{10}h-q + (11h-10) (t_1-0.1).
\end{equation}

Next, we show that for any possible value of $q$, there exists $t_1$ such that
either condition \eqref{eq:support1} or condition \eqref{eq:support2} is violated.
Observe that when $q$ increases, the lower bound in \eqref{eq:support1} increases,
while the lower bound in \eqref{eq:support2} decreases.
Thus, once we find a $q^*$ such that
\begin{enumerate}
\item[(a)] $\exists t_1'\in [0,1]$ with $q^*-q^* t_1' > U_h(t_1',0.1)$, and
\item[(b)] $\exists t_1'' \in [0.1,1]$ with $\frac{11}{10}h-q^* + (11h-10) (t_1''-0.1) > U_h(t_1'',0)$,
\end{enumerate}
then for any $q\geq q^*$, by \eqref{eq:support1} and (a), $\ucap^A(t_1',0.1) > U_h(t_1',0.1)$, a contradiction;
for any $q < q^*$, by \eqref{eq:support2} and (b), $\ucap^A(t_1'',0) > U_h(t_1'',0)$, again a contradiction.

We find that $q^* = 0.6979$ satisfies (a) and (b), with $t_1' = 0.26$ and $t_1'' = 0.32$.
\end{proof}

\section{A $0.67776$-Competitive Mechanism for Multiple Items}\label{sect:avg-PA}

Cole et al.~\cite{ColeGG2013-EC} introduced a family of SP mechanisms called \emph{Partial Allocation} (PA) mechanisms,
which work for multiple agents and multiple items. We describe the two-agent version $\pa_c$ below.\footnote{\label{fn:pa}Since
the utility functions are normalized, $0<u_2(\veca_2)^c,u_1(\veca_1)^{1/c}\leq 1$. So Step 2 of $\pa_c$ is legitimate.\\
Also, we note that the eventual utility attained by agent 1 is $u_1(\veca_1)\cdot u_2(\veca_2)^c = W(u_1,u_2)$,
and the eventual utility attained by agent 2 is $u_1(\veca_1)^{1/c}\cdot u_2(\veca_2) = W(u_1,u_2)^{1/c}$.}
In~\cite{ColeGG2013-AAMAS}, they showed that a variant of a PA mechanism, which we denote by \pamax,
is SP and it is $\frac 23$-competitive for two agents and multiple items.
We will show that by taking a suitable weighted average of two PA mechanisms and \pamax, we break the $\frac 23$ barrier.

\smallskip

\noindent
\begin{tabular}{|l|l|}
\hline
$\pa_c(\vecu_1,\vecu_2)$ for $0 < c < \infty$: & \pamax$(\vecu_1,\vecu_2)$:\\
1. Compute the feasible allocation $(\veca_1,\veca_2)$ & 1. Compute the allocation of $\pa_1(\vecu_1,\vecu_2)$.\\
~~~~that maximizes $u_1(\veca_1) \cdot u_2(\veca_2)^c$. & 2. Compute the allocation in which each\\
~~~~Let $W(u_1,u_2)$ denote the maximal value. & ~~~~item is split evenly among the two agents.\\
2. Agent $1$ is allocated a $u_2(\veca_2)^c$ fraction of $\veca_1$; & 3. Output the allocation that yields higher\\
~~~~agent $2$ is allocated a $u_1(\veca_1)^{1/c}$ fraction of $\veca_2$. & ~~~~social welfare.\\
~~~~(See footnote \ref{fn:pa}). & \\
\hline
\end{tabular}

%

\smallskip

To build up an intuition on why this might work, we look at an almost worst case scenario for \pamax,
where there are two items, $\vecu_1=(0.99,0.01)$ and $\vecu_2=(0.5,0.5)$.
In $\pa_1$, $\veca_1=(1,0)$ and $\veca_2=(0,1)$,
$u_1(\veca_1) = 0.99$ and $u_2(\veca_2) = 0.5$.
The eventual allocation to agent $1$ is $\frac 12 \veca_1$, i.e., reducing $\veca_1$ by half.
This reduction harms the eventual social welfare hugely.

If we consider $\pa_c$ for some $c$ less than $1$, say $c=0.5$,
the eventual allocation to agent $1$ is much better, which is $0.5^{0.5} \veca_1 \approx 0.707 a_1$.
The eventual allocation to agent $2$ in $\pa_{0.5}$ is $0.99^2 \veca_2$, which is slightly worse than $0.99 \veca_2$,
the eventual allocation to agent $2$ in $\pa_1$.
But overall, the social welfare in $\pa_{0.5}$ is much better.

However, there are bad scenarios for $\pa_{0.5}$, e.g., when the utility functions of the two agents in the last paragraph are swapped.
To attain an overall good competitiveness, we consider some weighted average of $\left(\frac 12\cdot \pa_c + \frac 12\cdot \pa_{1/c}\right)$ and \pamax.
For each choice of $c$ and the weights, by using the tool of AUR,
we can compute the competitive ratio with math software.
We find that for a suitable choice of $c$ and the weights, a competitive ratio of $0.67776$ is attained.

\begin{theorem}
The mechanism
$$\left(\frac{1029}{4000} \cdot \pa_{0.421} + \frac{1029}{4000} \cdot \pa_{1/0.421} + \frac{971}{2000}\cdot \text{\emph{\pamax}}\right)$$
is SP, and it is at least $0.67776$-competitive.
\end{theorem}

\begin{proof}
The mechanism is SP since it is a weighted average of three SP mechanisms.

For any $u_1,u_2$, let $(\us_1,\us_2)\in \aur(u_1,u_2)$ be a point that attains the optimal social welfare $\swopt$.
Observe that $(1,0),(0,1)\in \aur(u_1,u_2)$.
Thus by Proposition \ref{prop:aur-convex}, $\aur(u_1,u_2)$ contains the line segment $\ell_1(\us_1,\us_2)$ that joins $(1,0)$ and $(\us_1,\us_2)$,
and also the line segment $\ell_2(\us_1,\us_2)$ that joins $(0,1)$ and $(\us_1,\us_2)$.
Let $\ell(\us_1,\us_2) := \ell_1(\us_1,\us_2) \cup \ell_2(\us_1,\us_2)$.

Recall the notion $W(u_1,u_2)$ defined in $\pa_c$.
Note that $\SW^{\pa_c}(u_1,u_2) = W(u_1,u_2) + W(u_1,u_2)^{1/c}$, and
$$W(u_1,u_2) \geq \max_{(r_1,r_2)\in \ell(\us_1,\us_2)} r_1 \cdot (r_2)^c.$$
These allow us to turn to the easier single-variate optimization problems along the two line segments.

For our choices of $c$ and the weights, we run through all possible values of $(\us_1,\us_2)$,
which is the set $\{(\us_1,\us_2)\,|\,0\leq \us_1,\us_2\leq 1\text{ and }1\leq \us_1+\us_2\leq 2\}$, to compute the competitive ratio.
We perform this with math software, but first only considering discrete points where $\us_1,\us_2$ are multiples of $\frac{1}{2000}$.
The competitive ratio over these discrete points is at least $0.67844$.

For general $(\us_1,\us_2)$, we first round them down to the nearest multiple of $\frac{1}{2000}$ to $(\tilde{u}^*_1,\tilde{u}^*_2)$.
Observe that
$$\max_{(r_1,r_2)\in \ell(\us_1,\us_2)} r_1 \cdot (r_2)^c \geq \max_{(r_1,r_2)\in \ell(\tilde{u}^*_1,\tilde{u}^*_2)} r_1\cdot (r_2)^c,$$
and $1\leq \us_1 + \us_2 \leq \tilde{u}^*_1 + \tilde{u}^*_2 + \frac{1}{1000}$.
Thus, the competitive ratio over general points is at least $0.67844\times \left(1-\frac{1}{1000}\right) > 0.67776$.
\end{proof}

\section{Dynamic-Increasing-Price Mechanisms}\label{sect:DIP}

In this section, we propose a sub-class of SP mechanisms for any number of agents and items.
For simplicity, we describe the general form of DIP mechanisms for the case with two agents.
It is easy to see how to generalize to the cases with any number of agents.

Recall that the number of items is $m$.
For each agent $i$, she has one unit of \emph{virtual money}.
For each item $j$, agent $i$ will be given a price function $P_j^{u_{3-i}}: [0,1]\ra \rr^+ \cup \{0,+\infty\}$,
which is an increasing function that depends on the utility function of the agent $3-i$.
We will write $P_j^{-i}$ as a shorthand for $P_j^{u_{3-i}}$.
The value $P_j^{-i}(y)$ describes the marginal price when agent $i$ has already purchased $y$ units of item $j$.
In other words, if agent $i$ purchases $x_j$ units of item $j$,
she needs to pay $T_j^{-i}(x_j) := \int_0^{x_j} P_j^{-i}(y)\,dy$ unit of her virtual money.
Agent $i$ will purchase an allocation $\vec{x} = (x_1,x_2,\cdots,x_m)$ which is in
$$\argmax_{\substack{\forall j,~0\leq x_j\leq 1 \\ \sum_{j=1}^m T_j^{-i}(x_j) \leq 1}}\,\,\,\,\sum_{j=1}^m \uij x_j.$$

A DIP mechanism is obviously SP, since each agent is accepting prices which she herself cannot influence, and uses them to decide an optimal purchase.
Yet, feasibility is a delicate issue, particularly when $m$ is large.

While there are some similarities between DIP and IP mechanisms,
we note that there is no process of dictator-swapping in DIP, and DIP is not swap-dictatorial and not full in general.
Another difference between DIP and IP is that in DIP the price functions are \emph{price versus quantity of item already purchased},
while in IP they are \emph{price versus virtual money already spent}.

DIP is well motivated, as explained below.
In the first-best allocation, if $u_{1j} > u_{2j}$, then agent $1$ gets all of item $j$.
When $u_{1j}$ is high but $u_{2j}$ is low, ideally we want to construct the price function $P_j^{-1},P_j^{-2}$ such that 
$P_j^{-2}$ is lower to encourage agent $1$ to purchase more item $j$,
but $P_j^{-1}$ is higher to discourage agent $2$ from purchasing more item $j$.
Adjusting the prices dynamically in this manner might help pushing the allocation in DIP towards the first best allocation,
and therefore might be hopeful to attain better competitiveness.

The way that DIP works is similar to \emph{price taking} in the context of Fisher market:
suppose an agent enters a Fisher market where items are sold, the prices she takes surely depend on other agents' preferences,
while she herself can also have influence on the prices --- as a mechanism, DIP deliberately removes such influence by herself.
In a non-rigorous sense, we may say that DIP is a \emph{rescue} for Fisher markets from being non-strategyproof.

We note that Cole et al.~\cite[Theorem 1]{ColeGG2013-AAMAS} proved that when there are two agents and multiple items,
Fisher market equilibrium allocation (which is also the \emph{proportional fair} allocation) is highly competitive against the first-best allocation.

As a simple example, the SP mechanism that allocates half of each item to both agents is DIP,
by setting $P_j^{-i}(y) = 0$ for $y\in\left[0,\frac 12\right]$ and $P_j^{-i}(y) = +\infty$ for $y\in\left[\frac 12,1\right]$.

\begin{prop}
The $\frac 56$-competitive mechanism given in Section \ref{sect:fully-two-items} is a DIP mechanism.
\end{prop}

We give the price functions for agent $1$ on item $1$ below; other price functions are defined symmetrically.
Recall that $\vecu_2 = (t_2,1-t_2)$, and also the definition of $f$ in \eqref{eq:f-explicit}.
Let $\tau=\frac 12 - f(t_2)$. Then
$$P_1^{-1} (y) :=
\begin{cases}
0, & y\in \left[0,\tau\right]\\
C, & y\in \left(\tau,f(\frac 12) + \tau\right]\\
\frac{C}{g(y)} - C, & y\in \left(f(\frac 12) + \tau, \frac 12 + \tau\right]\\
+\infty,& y\in \left(\frac 12 + \tau, 1\right],
\end{cases}$$
where $g(y)$ denotes the unique value of $z\in \left[\frac 15,\frac 12\right]$ such that $f(1-z) - f(t_2) + \frac 12 = y$,
and $C$ is the positive constant such that $\int_0^{1-f(t_2)}P_1^{-1} (y)\,dy = 1$.


\section{Discussion and Open Problems}

The most important problem for future research is to seek optimal competitive mechanisms.
As we have already seen, even for the case with two items, where the setting is essentially single-parameter,
the use of analytical tools seems unavoidable --- for instance,
it looks unlikely to have a \emph{natural} interpretation of the $\frac 56$-competitive mechanism,
which we believe to be an optimal full mechanism.
In the study of revenue-optimal mechanisms with prior, more advanced analytical tools,
including duality theory and variational calculus, have played key roles. See, e.g., \cite{PP2011,GK2014,GK2015}.
We believe that such tools will be useful for our problem too;
for instance, duality theory is likely to be useful for showing that the upper bound for full mechanism is \emph{exactly} $\frac 56$.

In Section \ref{sect:human-upper-bound}, we use Rochet's characterization to prove a non-trivial upper bound.
The proof only considers the restrictions to two cross sections of $\ucap^A$,
so clearly it has not yet fully exploited the power of the characterization.
An interesting research agenda is to seek a more sophisticated use of the characterization for proving better upper bounds,
for either the case with two items, or for those with more items.

\section*{Acknowledgments}

The author thanks the anonymous reviewers for their thoughtful suggestions which help improving the structure and presentation of this paper.


\bibliographystyle{alpha}
\bibliography{literature}

\end{document}